\documentclass[aps,pra,reprint,superscriptaddress,longbibliography]{revtex4-2}

\usepackage[T1]{fontenc}
\usepackage{lmodern}
\usepackage{amsmath,amssymb,amsthm,mathtools,bm}
\usepackage{physics}
\usepackage{braket}
\usepackage{hyperref}
\hypersetup{colorlinks=true,citecolor=blue,linkcolor=blue,urlcolor=blue}

\newtheorem{theorem}{Theorem}
\newtheorem{proposition}{Proposition}
\newtheorem{lemma}{Lemma}
\newtheorem{corollary}{Corollary}
\theoremstyle{remark}
\newtheorem{remark}{Remark}

\begin{document}
	
	\title{Trace-to-Hilbert-Schmidt Speed Ratio in Quantum Dynamics: Universal Bounds and Effective Rank}
	
	\author{Hossein Rangani Jahromi\\
		\textit{Physics Department, Faculty of Sciences,
			Jahrom University, P.B. 74135111, Jahrom, Iran}}

	\begin{abstract}
		We address the ratio between the trace speed and the Hilbert-Schmidt speed for differentiable finite-dimensional quantum states, $\mathcal R=\|\partial_\phi\rho\|_1/\|\partial_\phi\rho\|_2$. Since the tangent $\partial_\phi\rho$ is Hermitian and traceless, $\mathcal R$ obeys stronger bounds than those for generic operators. For any nonzero tangent of rank $r$, we prove the sharp bounds $\sqrt{2}\le\mathcal R\le\sqrt{r}$ for even $r$ and $\sqrt{2}\le\mathcal R\le\sqrt{r-1/r}$ for odd $r$, characterizing all equality cases. Nonstationary pure-state and qubit families saturate the lower bound $\mathcal R=\sqrt2$. For odd Hilbert-space dimension $d$, we further prove the sharp global maximum $\mathcal R\le\sqrt{d-1/d}$. Interpreting $\mathcal R^2$ as the participation ratio of the singular-value distribution yields an effective-rank picture, $r_{\mathrm{eff}}=\mathcal R^2$. We decompose the effective rank into classical eigenvalue and quantum eigenvector contributions and obtain the bound $r_{\mathrm{eff}}\le r_C+r_Q$, with equality when either component vanishes. Linking the effective rank to the quantum Fisher information $F$ gives $r_{\mathrm{eff}}\ge 8\,\mathrm{TS}^2/F$, showing that a large effective rank is required when this lower bound substantially exceeds the universal minimum value $2$. Finally, a hierarchy of quantum speed limits shows how the effective rank controls the tightness of bounds expressed through the Hilbert-Schmidt speed.
	\end{abstract}
	
	\maketitle
	
	\section{Introduction}
	
	Local speeds quantify infinitesimal changes of quantum states and, depending on the underlying norm or metric, provide distinct geometric or statistical notions of local distinguishability relevant to quantum metrology, quantum geometry, and quantum speed limits \cite{BraunsteinCaves1994,Paris2009,GessnerSmerzi2018,DeffnerCampbell2017,PoggiCampbellDeffner2021}. Different choices of norm or metric lead to inequivalent notions of statistical speed, each emphasizing a different aspect of the underlying dynamics. In particular, the trace norm and the Hilbert--Schmidt norm give rise to two norm-induced local speeds for a differentiable family of density operators $\rho(\phi)$:
	\[
	\mathrm{TS}(\phi)=\frac{1}{2}\|\partial_\phi\rho(\phi)\|_1,
	\qquad
	\mathrm{HSS}(\phi)=\frac{1}{2}\|\partial_\phi\rho(\phi)\|_2.
	\]
\par
	
	The trace-norm speed, which we refer to as the trace speed (TS), is directly associated with the trace distance, the standard operational measure of state distinguishability. The Hilbert-Schmidt speed (HSS) is particularly attractive because it can be evaluated with much less computational effort than the trace speed, while still capturing essential dynamical features. These complementary strengths have made both speeds valuable tools: they have been successfully applied as effective witnesses of non-Markovianity and as figures of merit for phase estimation in multi-qubit systems \cite{Jahromi2020HSSNM,JahromiLoFranco2021,Alizadeh2025}.
	\par
	Despite the broad use of norm-based speeds, the relation between trace and Hilbert--Schmidt local speeds has not, to our knowledge, been systematically characterized at the level of density-operator tangents. This point is crucial because the tangent operator
	\[
	X(\phi)=\partial_\phi\rho(\phi)
	\]
	is not a generic matrix: it is always Hermitian and traceless. These elementary constraints strongly restrict its singular-value structure. As a result, norm inequalities that are optimal for arbitrary operators need not be optimal for physically admissible quantum-state tangents. The aim of this work is to exploit this structure and derive sharp, state-independent constraints on the ratio between the two local speeds.
	
	We therefore study the dimensionless quantity
	\[
	\mathcal R(\phi)=\frac{\mathrm{TS}(\phi)}{\mathrm{HSS}(\phi)}
	=\frac{\|\partial_\phi\rho(\phi)\|_1}{\|\partial_\phi\rho(\phi)\|_2}.
	\]
	Since $X(\phi)$ is Hermitian, its singular values are the absolute values of its eigenvalues. Hence $\mathcal R(\phi)$ depends only on the singular-value profile of the instantaneous tangent operator and not on its overall scale. In this sense, $\mathcal R$ isolates the spectral structure of the local motion.
	
	Our main result establishes that the Hermitian traceless structure of density-operator tangents enforces sharp universal bounds on this ratio. For every nonzero tangent operator $X=\partial_\phi\rho$ of rank $r$, we prove
	\[
	\sqrt{2}\le \frac{\|X\|_1}{\|X\|_2}\le
	\begin{cases}
		\sqrt r, & r\ \text{even},\\[6pt]
		\sqrt{\,r-\dfrac{1}{r}\,}, & r\ \text{odd},
	\end{cases}
	\]
	and we completely characterize all equality cases. The lower bound is strictly stronger than the generic operator bound $1\le \|X\|_1/\|X\|_2$ and arises from the balance between positive and negative spectral weight imposed by tracelessness. The upper bound is saturated exactly when all nonzero singular values are equal for even rank, or when the positive and negative eigenvalues are each constant and have multiplicities differing by one for odd rank. We further show that every nonstationary pure-state family has rank-two tangent and therefore satisfies $\mathcal R=\sqrt{2}$, and that the same rigidity holds for all nonstationary qubit families. In odd Hilbert-space dimension $d$, we obtain the sharper global bound
	\[
	\mathcal R\le \sqrt{d-\frac{1}{d}},
	\]
	which shows that the algebraic ceiling $\sqrt d$ allowed for a generic full-rank operator is unattainable for traceless Hermitian tangents when $d$ is odd.
	
	Beyond these sharp inequalities, the ratio $\mathcal R$ admits a natural interpretation in terms of the distribution of singular values. Indeed, if $\sigma_i$ are the nonzero singular values of $X$, then
	\[
	\mathcal R^2=\frac{\bigl(\sum_i \sigma_i\bigr)^2}{\sum_i \sigma_i^2}
	\]
	is the participation ratio of the normalized singular-value distribution. This suggests interpreting $\mathcal R^2$ as an effective rank of the local dynamics, namely, as a measure of how many singular modes contribute appreciably to the instantaneous motion. From this viewpoint, the universal lower bound means that every nonstationary quantum-state tangent necessarily involves at least two effective modes, while larger values of $\mathcal R$ indicate more broadly distributed tangent spectra.
	
	This effective-rank viewpoint also leads to direct physical consequences. First, by decomposing the tangent operator into diagonal and off-diagonal parts in the instantaneous eigenbasis of $\rho(\phi)$, we separate classical contributions associated with eigenvalue variation from quantum contributions associated with eigenvector variation. Second, in the case of unitary parameter encoding, we derive an inequality relating the effective rank to the quantum Fisher information (QFI), the fundamental quantity governing the ultimate precision of parameter estimation \cite{BraunsteinCaves1994,Paris2009,Helstrom1976,Holevo1982}. This shows that, at fixed trace speed, smaller QFI increases the lower bound on the effective rank, and can force it to be correspondingly larger when that bound rises substantially above the universal minimum. Third, because quantum speed limits are controlled by local rates of change and path lengths in state space \cite{MandelstamTamm1945,GiovannettiLloydMaccone2003,Taddei2013,delCampo2013,DeffnerLutz2013,DeffnerCampbell2017,PoggiCampbellDeffner2021}, the identity
	\[
	\mathrm{TS}=\sqrt{r_{\mathrm{eff}}}\,\mathrm{HSS}
	\]
	provides a direct comparison between trace-speed-based and Hilbert--Schmidt-speed-based speed-limit estimates, with the gap controlled exactly by the effective rank $r_{\mathrm{eff}}=\mathcal R^2$.
	
	The paper is organized as follows. Section~\ref{sec:preliminaries} introduces notation and basic structural facts about density-operator tangents. Section~\ref{sec:bounds} proves the universal bounds on $\mathcal R$ and characterizes the saturation conditions. Section~\ref{sec:pure} treats pure-state dynamics and the unitary specialization. Section~\ref{sec:effective} develops the effective-rank interpretation and discusses the qubit case. Section~\ref{sec:consequences} derives the decomposition into classical and quantum contributions, the relation to the QFI in the unitary setting, and the speed-limit comparison. Section~\ref{sec:odd} proves the sharp odd-dimensional maximum. Section~\ref{sec:discussion} concludes.

	\section{Preliminaries and realizability of tangent operators}
	\label{sec:preliminaries}
	
	Let $\mathcal H\simeq\mathbb C^d$ be a finite-dimensional Hilbert space, and let $\rho(\phi)$ be a $C^1$ family of density operators depending on a real parameter $\phi$. Thus, for every $\phi$,
	\begin{equation}
		\rho(\phi)^\dagger=\rho(\phi), \qquad \rho(\phi)\ge0, \qquad \Tr \rho(\phi)=1.
	\end{equation}
	The tangent operator is defined as
	\begin{equation}
		X(\phi):=\partial_\phi\rho(\phi).
	\end{equation}
	Differentiating the Hermiticity and trace conditions gives
	\begin{equation}
		X(\phi)^\dagger=X(\phi), \qquad \Tr X(\phi)=0.
		\label{eq:HermTraceless}
	\end{equation}
	Hence $X(\phi)$ is always Hermitian and traceless. At stationary points, where $X(\phi)=0$, both speeds vanish and the ratio $\mathcal R$ is undefined. In all that follows we restrict attention to nonstationary points, i.e., $X(\phi)\neq0$.
	
	The key structural observation is that tangent operators of density matrices occupy a much smaller class than arbitrary matrices. This simple observation is the source of the strengthened lower bound derived below.
	
	For any operator $A$, we denote its singular values (including zeros when present) by $\sigma_i(A)$. The Schatten norms are then
	\begin{equation}
		\norm{A}_1 = \sum_i \sigma_i(A), \qquad
		\norm{A}_2 = \sqrt{\sum_i \sigma_i(A)^2}.
	\end{equation}
	When convenient, we sum only over the $r=\rank(A)$ nonzero singular values. Because $X(\phi)$ is Hermitian, its singular values are the absolute values of its eigenvalues. If $\lambda_1,\dots,\lambda_d$ are the eigenvalues of $X$, then
	\begin{equation}
		\norm{X}_1 = \sum_{i=1}^d |\lambda_i|, \qquad
		\norm{X}_2 = \sqrt{\sum_{i=1}^d \lambda_i^2},
	\end{equation}
	and therefore
	\begin{equation}
		\mathcal R(X)^2 = \frac{\bigl(\sum_{i=1}^d |\lambda_i|\bigr)^2}{\sum_{i=1}^d \lambda_i^2}.
		\label{eq:Rlambda}
	\end{equation}
	Thus the problem reduces completely to the study of nonzero traceless Hermitian matrices.
	
	A useful structural fact is that every traceless Hermitian operator can be realized locally as the tangent of a legitimate density-operator family.
	
	\begin{lemma}
		\label{lem:realization}
		Let $X$ be any traceless Hermitian operator on $\mathcal H$. Then there exist a full-rank density operator $\rho_0$ and $\varepsilon>0$ such that
		\begin{equation}
			\rho(\phi)=\rho_0+\phi X
		\end{equation}
		is a density operator for all $|\phi|<\varepsilon$. In particular,
		\begin{equation}
			\partial_\phi\rho(\phi)\big|_{\phi=0}=X.
		\end{equation}
	\end{lemma}
	
	\begin{proof}
		If $X=0$ the statement is trivial; we assume $X\neq0$ in the following. Choose any full-rank density operator $\rho_0>0$. Let $\mu_{\min}>0$ be its smallest eigenvalue. Because $X$ is Hermitian, its operator norm $\norm{X}_\infty$ is finite. (Recall that for Hermitian operators $\norm{X}_\infty = \max_i |\lambda_i(X)|$, and we have $-\norm{X}_\infty I \le X \le \norm{X}_\infty I$, which implies $-|\phi|\,\norm{X}_\infty I \le \phi X \le |\phi|\,\norm{X}_\infty I$ for all real $\phi$.) For $|\phi|<\mu_{\min}/\norm{X}_\infty$, we have
		\begin{equation}
			\rho_0+\phi X \ge \rho_0 - |\phi|\,\norm{X}_\infty I \ge \mu_{\min}I - |\phi|\,\norm{X}_\infty I \ge 0.
		\end{equation}
		Moreover, $\Tr(\rho_0+\phi X) = \Tr\rho_0 + \phi\,\Tr X = 1$, since $\Tr X=0$. Hermiticity is obvious. Hence $\rho(\phi)$ is a density operator for sufficiently small $|\phi|$, and its derivative at $\phi=0$ is $X$.
	\end{proof}
	
	\section{Universal bounds}
	\label{sec:bounds}
	
	We now derive the sharp general bounds on $\mathcal R$, refined by the parity of the rank.
	
	\begin{theorem}
		\label{thm:universal}
		Let $X\neq0$ be a traceless Hermitian operator of rank $r$. Then
		\[
		\sqrt{2}\le \frac{\|X\|_1}{\|X\|_2}\le
		\begin{cases}
			\sqrt r, & r\ \text{even},\\[8pt]
			\sqrt{\,r-\dfrac{1}{r}\,}, & r\ \text{odd}.
		\end{cases}
		\]
		Both the lower bound and the upper bounds are sharp.
		\begin{itemize}
			\item The lower bound holds with equality if and only if $X$ has exactly two nonzero eigenvalues $\{\lambda,-\lambda\}$ (rank two).
			\item For even $r$, the upper bound is attained if and only if all nonzero singular values are equal; equivalently, the $r$ nonzero eigenvalues consist of $r/2$ copies of $+s$ and $r/2$ copies of $-s$.
			\item For odd $r$, the upper bound is attained if and only if the positive eigenvalues are all equal, the negative eigenvalues are all equal, and their multiplicities are $(r-1)/2$ and $(r+1)/2$, up to exchange.
		\end{itemize}
	\end{theorem}
	
	\begin{proof}
		Since $X$ is Hermitian and traceless, its nonzero eigenvalues can be written as
		\begin{equation}
			\alpha_1,\dots,\alpha_m,\,-\beta_1,\dots,-\beta_n,
			\label{eq:specsplit}
		\end{equation}
		where $\alpha_k>0$, $\beta_j>0$, $m\ge 1$, $n\ge 1$, and $m+n=r$. Tracelessness implies
		\begin{equation}
			\sum_{k=1}^{m}\alpha_k = \sum_{j=1}^{n}\beta_j =: Q > 0.
			\label{eq:Q}
		\end{equation}
		Consequently,
		\begin{equation}
			\norm{X}_1 = 2Q, \qquad
			\norm{X}_2^2 = \sum_{k=1}^{m}\alpha_k^2 + \sum_{j=1}^{n}\beta_j^2.
			\label{eq:normssplit}
		\end{equation}
		Thus
		\begin{equation}
			\mathcal R(X)^2 = \frac{4Q^2}{\sum_{k=1}^{m}\alpha_k^2 + \sum_{j=1}^{n}\beta_j^2}.
			\label{eq:R2split}
		\end{equation}
		
		We first prove the lower bound. By the elementary estimate
		\begin{equation}
			\sum_{k=1}^{m}\alpha_k^2 \le \biggl(\sum_{k=1}^{m}\alpha_k\biggr)^2 = Q^2,
		\end{equation}
		with equality if and only if $m=1$. Similarly,
		\begin{equation}
			\sum_{j=1}^{n}\beta_j^2 \le Q^2,
		\end{equation}
		with equality if and only if $n=1$. Substituting into Eq.~\eqref{eq:R2split} yields
		\begin{equation}
			\mathcal R(X)^2 \ge \frac{4Q^2}{2Q^2} = 2,
		\end{equation}
		hence $\mathcal R(X)\ge\sqrt{2}$. Equality holds if and only if $m=n=1$, i.e., $X$ has exactly one positive and one negative eigenvalue, so $\rank(X)=2$. Conversely, any traceless Hermitian rank-two operator has spectrum $\{\lambda,-\lambda\}$ with $\lambda\neq0$, and therefore $\mathcal R(X)=\sqrt2$. This proves the lower bound and its equality condition.
		
		For the upper bound we minimize the denominator for fixed $Q$, $m$, $n$. By Cauchy--Schwarz,
		\begin{equation}
			\sum_{k=1}^{m}\alpha_k^2 \ge \frac{Q^2}{m},\qquad
			\sum_{j=1}^{n}\beta_j^2 \ge \frac{Q^2}{n},
		\end{equation}
		with equalities if and only if all $\alpha_k$ are equal and all $\beta_j$ are equal. Hence
		\begin{equation}
			\mathcal R(X)^2 \le \frac{4Q^2}{Q^2(1/m+1/n)} = \frac{4mn}{m+n}.
			\label{eq:mn_bound}
		\end{equation}
		
		For fixed rank $r=m+n$, the function $f(m)=4m(r-m)/r$ is maximized when $m$ is as close as possible to $r/2$. A direct calculation gives
		\begin{equation}
			\max_{\substack{m,n\ge1\\m+n=r}} \frac{4mn}{m+n}
			= \begin{cases}
				r, & r\ \text{even},\\[6pt]
				r-\dfrac{1}{r}, & r\ \text{odd}.
			\end{cases}
			\label{eq:fixedrmax}
		\end{equation}
		The maximum for even $r$ is achieved by $m=n=r/2$; for odd $r$ it is achieved by $\{m,n\}=\{(r-1)/2,\,(r+1)/2\}$. Substituting into Eq.~\eqref{eq:mn_bound} gives the stated piecewise upper bound. The equality conditions follow from the chain of inequalities: for even $r$, all $\alpha_k$ and all $\beta_j$ must be equal, which together with $m=n$ yields equal multiplicities of $\pm s$; for odd $r$, the same constancy holds but the multiplicities differ by one. Conversely, any operator with such a spectrum saturates the bound. This completes the proof.
	\end{proof}
	
	\begin{remark}
		The generic bound $\|X\|_1/\|X\|_2\le\sqrt r$ (valid for any operator) is still true for odd $r$, but it is not sharp because tracelessness prevents equal singular values. The refined bound $\sqrt{r-1/r}$ is the sharp restriction. The lower bound is strictly stronger than the trivial $1$ and is a direct consequence of tracelessness.
	\end{remark}
	
	\section{Pure-state families}
	\label{sec:pure}
	
	Pure-state families provide a rigid and physically important special case in which the universal lower bound is always saturated.
	
	\begin{proposition}
		\label{prop:pure}
		Let $\ket{\psi(\phi)}$ be a differentiable normalized state vector and $\rho(\phi)=\ket{\psi(\phi)}\bra{\psi(\phi)}$. If $\partial_\phi\rho(\phi)\neq0$, then $\partial_\phi\rho(\phi)$ has rank $2$ and
		\begin{equation}
			\mathcal R(\phi)=\sqrt2.
		\end{equation}
	\end{proposition}
	
	\begin{proof}
		Write $\ket{\dot\psi}:=\partial_\phi\ket{\psi}$. Differentiating the projector gives
		\begin{equation}
			\partial_\phi\rho = \ket{\dot\psi}\bra{\psi} + \ket{\psi}\bra{\dot\psi}.
			\label{eq:drho_pure}
		\end{equation}
		Define $\alpha:=\braket{\psi|\dot\psi}$. From $\braket{\psi|\psi}=1$ we obtain $\alpha+\alpha^*=0$, so $\alpha$ is purely imaginary. Decompose
		\begin{equation}
			\ket{\dot\psi} = \alpha\ket{\psi} + \ket{\eta},
			\qquad \braket{\psi|\eta}=0.
			\label{eq:decomp}
		\end{equation}
		Substituting Eq.~\eqref{eq:decomp} into Eq.~\eqref{eq:drho_pure} gives
		\begin{equation}
			\partial_\phi\rho = (\alpha+\alpha^*)\ket{\psi}\bra{\psi} + \ket{\eta}\bra{\psi} + \ket{\psi}\bra{\eta}
			= \ket{\eta}\bra{\psi} + \ket{\psi}\bra{\eta},
			\label{eq:pure_eta}
		\end{equation}
		since $\alpha+\alpha^*=0$. If $\ket{\eta}=0$, then $\partial_\phi\rho=0$, contradicting nonstationarity. Hence $\beta:=\norm{\eta}>0$, and we define $\ket{\psi_\perp}:=\ket{\eta}/\beta$. Then
		\begin{equation}
			\partial_\phi\rho = \beta\bigl(\ket{\psi_\perp}\bra{\psi} + \ket{\psi}\bra{\psi_\perp}\bigr).
		\end{equation}
		In the orthonormal basis $\{\ket{\psi},\ket{\psi_\perp}\}$ the matrix is $\begin{pmatrix}0&\beta\\\beta&0\end{pmatrix}$, with eigenvalues $\pm\beta$. Thus $\partial_\phi\rho$ acts nontrivially only on the two-dimensional subspace spanned by $\ket{\psi}$ and $\ket{\psi_\perp}$; its rank is $2$ and its nonzero singular values are both $\beta$. Therefore $\norm{\partial_\phi\rho}_1=2\beta$, $\norm{\partial_\phi\rho}_2=\sqrt2\,\beta$, and $\mathcal R=\sqrt2$.
	\end{proof}
	
	The unitary specialization is immediate.
	
	\begin{corollary}
		\label{cor:unitary}
		If $\partial_\phi\ket{\psi(\phi)} = -iG\ket{\psi(\phi)}$ for a Hermitian generator $G$, then at every nonstationary point
		\begin{equation}
			\mathrm{TS}(\phi)=\Delta G,\qquad
			\mathrm{HSS}(\phi)=\frac{\Delta G}{\sqrt2},\qquad
			\mathcal R(\phi)=\sqrt2,
		\end{equation}
		where $\Delta G = \sqrt{\expval{G^2}-\expval{G}^2}$.
	\end{corollary}
	
	\begin{proof}
		The orthogonal component is $|\eta\rangle = -i(G-\expval{G})|\psi\rangle$, whose norm is $\|\eta\| = \sqrt{\langle\psi|(G-\langle G\rangle)^2|\psi\rangle} = \Delta G$ by definition of the variance. The statement follows from Proposition~\ref{prop:pure}. Thus, for unitary pure-state dynamics generated by $G$, the trace speed and Hilbert-Schmidt speed are both controlled by the generator variance, and their ratio remains fixed at $\sqrt2$.
	\end{proof}
	
	\section{Effective rank and the qubit case}
	\label{sec:effective}
	
	The ratio $\mathcal R$ admits a direct information-theoretic interpretation in terms of the singular-value distribution of the tangent operator. Let $\sigma_1,\dots,\sigma_r>0$ be the nonzero singular values of $X=\partial_\phi\rho$. Then
	\begin{equation}
		\mathcal R^2 = \frac{\bigl(\sum_{i=1}^r \sigma_i\bigr)^2}{\sum_{i=1}^r \sigma_i^2}.
	\end{equation}
	If $p_i = \sigma_i/\sum_j\sigma_j$, then $\mathcal R^2=1/\sum_i p_i^2$, which is the participation ratio (the inverse of the standard inverse participation ratio) of the normalized singular-value distribution. This prompts the definition
	\begin{equation}
		r_{\mathrm{eff}} := \frac{\bigl(\sum_{i=1}^r \sigma_i\bigr)^2}{\sum_{i=1}^r \sigma_i^2} = \mathcal R^2,
		\label{eq:reff}
	\end{equation}
	which we term the \textit{Schatten effective rank} of the tangent operator.
	
	Theorem~\ref{thm:universal} implies the bounds
	\begin{equation}
		2 \le r_{\mathrm{eff}} \le
		\begin{cases}
			r, & r\ \text{even},\\[4pt]
			r-\dfrac{1}{r}, & r\ \text{odd},
		\end{cases}
		\label{eq:reffbounds}
	\end{equation}
	and in all cases $r_{\mathrm{eff}}\le r\le d$. The lower endpoint corresponds precisely to rank-two tangents, while the upper endpoint for even $r$ occurs exactly when all nonzero singular values are equal; for odd $r$ the maximum is slightly below $r$, realized by the extremal multiplicities $(r\pm1)/2$ with constant values on each support.
	
	Thus $\mathcal R=\sqrt2$ certifies only that the local tangent has effective rank two, regardless of whether the state is pure or mixed. Similarly, $\mathcal R>\sqrt2$ signals participation of more than two singular modes but does not identify the underlying physical mechanism.
	
	For qubit systems the ratio is completely rigid, independent of the underlying state family.
	
	\begin{proposition}
		\label{prop:qubit}
		For every differentiable family of qubit states, at every nonstationary point one has $\mathcal R=\sqrt2$.
	\end{proposition}
	
	\begin{proof}
		Any nonzero traceless Hermitian $2\times2$ matrix has eigenvalues $\pm\lambda$ for some $\lambda>0$. Its singular values are both $\lambda$, giving $\mathcal R = 2\lambda/\sqrt{2\lambda^2} = \sqrt2$. This argument uses only the dimension and the traceless Hermitian structure, so the rigidity holds irrespective of whether the dynamics are unitary, dissipative, or non-Markovian. Hence every nonstationary differentiable qubit evolution satisfies $\mathcal R=\sqrt2$.
	\end{proof}
	
	\section{Physical consequences: decomposition, metrology, and speed limits}
	\label{sec:consequences}
	
	The effective rank $r_{\mathrm{eff}} = \mathcal R^2$ and the associated tangent speeds carry important physical information about the local dynamics. In this section, we derive three direct consequences of the framework established above: a decomposition into classical and quantum contributions, a tight relationship with the QFI, and a family of quantum speed limits governed by the effective rank.
	
	\subsection{Classical and quantum contributions to the effective rank}
	\label{sec:decomp}
	
	Let $\rho(\phi)$ be a differentiable family of density operators. On any open interval where the spectral multiplicities remain constant, the distinct eigenvalues $\lambda_a(\phi)$ and the corresponding spectral projectors $P_a(\phi)$ are $C^1$. We can then define the canonical decomposition
	\begin{equation}
		D(\phi) := \sum_a \bigl(\partial_\phi \lambda_a(\phi)\bigr) P_a(\phi), \qquad
		O(\phi) := \partial_\phi\rho(\phi) - D(\phi).
		\label{eq:intrinsic_DO}
	\end{equation}
	In an eigenbasis adapted to the spectral subspaces, $D$ is diagonal and $O$ is off-diagonal with respect to the spectral-block decomposition; it has no matrix elements within a given degenerate eigenspace block. They are orthogonal in the Hilbert-Schmidt inner product,
	\[
	\Tr(D^\dagger O)=0,
	\]
	so that $\|X\|_2^2 = \|D\|_2^2 + \|O\|_2^2$. At isolated points where eigenvalues cross, the decomposition can be performed relative to any diagonalizing basis; the inequality proved below remains valid for any such choice, although the attribution of $D$ to eigenvalue changes and $O$ to eigenvector rotations is then basis-dependent.
	
	For convenience, one may work in any instantaneous eigenbasis and set
	\[
	D_{ij} = \delta_{ij}\,\partial_\phi p_i, \qquad
	O_{ij} = (1-\delta_{ij}) X_{ij},
	\]
	which reproduces the projector form on intervals of constant multiplicity. Both $D$ and $O$ are Hermitian and traceless.
	
	We define the \textit{classical effective rank} as
	\[
	r_C := \begin{cases}
		\displaystyle\frac{\norm{D}_1^2}{\norm{D}_2^2}, & D\neq0,\\[10pt]
		0, & D=0,
	\end{cases}
	\]
	and similarly the \textit{quantum effective rank} $r_Q$ for $O$. The total effective rank is $r_{\mathrm{eff}} = \mathcal R^2 = \norm{X}_1^2/\norm{X}_2^2$.
	
	\begin{proposition}
		\label{prop:decomp}
		For any differentiable family of states, using any instantaneous eigenbasis to define $D$ and $O$,
		\[
		r_{\mathrm{eff}} \le r_C + r_Q.
		\]
		Equality is guaranteed when $D=0$ or $O=0$; a full characterization of equality conditions is not pursued here.
	\end{proposition}
	
	\begin{proof}
		By the triangle inequality for the trace norm, $\norm{X}_1 = \norm{D+O}_1 \le \norm{D}_1 + \norm{O}_1$. Combining this with the Hilbert-Schmidt orthogonality gives
		\[
		r_{\mathrm{eff}} \le \frac{(\norm{D}_1+\norm{O}_1)^2}{\norm{D}_2^2+\norm{O}_2^2}.
		\]
		For any non-negative numbers $a,b,c,d$ with $c,d>0$ (and taking the obvious limit if one denominator is zero), $\frac{(a+b)^2}{c+d} \le \frac{a^2}{c} + \frac{b^2}{d}$ because $(ad-bc)^2\ge0$. Applying this with $a=\|D\|_1$, $c=\|D\|_2^2$, etc., yields
		\[
		r_{\mathrm{eff}} \le \frac{\|D\|_1^2}{\|D\|_2^2} + \frac{\|O\|_1^2}{\|O\|_2^2} = r_C + r_Q.
		\]
		If $D=0$, then $r_C = 0$ and the triangle inequality becomes $\|X\|_1 = \|O\|_1$, so the chain of inequalities reduces to equalities; hence $r_{\mathrm{eff}} = r_Q = r_C + r_Q$. The same argument holds if $O=0$. Therefore equality is guaranteed whenever one of the two components vanishes.
	\end{proof}
	
	\begin{remark}
		This decomposition sharply distinguishes unitary and classical dynamics. For unitary families $\rho(\phi)=e^{-i\phi G}\rho_0 e^{i\phi G}$, the eigenvalues $p_i$ are constant, so $D=0$ and $r_{\mathrm{eff}} = r_Q$: the speed ratio is purely quantum. For families that remain diagonal in a fixed basis (i.e., dynamics equivalent to a classical probability evolution in a fixed basis), $O=0$ and $r_{\mathrm{eff}} = r_C$. In general open-system dynamics, both contributions are present, and their sum upper-bounds the total effective rank, with equality guaranteed when one of them vanishes.
	\end{remark}
	
	\subsection{Relation to the Quantum Fisher Information}
	\label{sec:qfi}
	
	Consider a unitary family $\partial_\phi\rho = -i[G,\rho]$. In the eigenbasis of $\rho=\sum_i p_i\ket{i}\bra{i}$, the standard expression for the symmetric logarithmic derivative quantum Fisher information (QFI) \cite{BraunsteinCaves1994} reads
	\[
	F = 2\sum_{i,j:\,p_i+p_j>0} \frac{(p_i-p_j)^2}{p_i+p_j}|G_{ij}|^2,
	\]
	while the squared Hilbert-Schmidt speed takes the form
	\[
	\mathrm{HSS}^2 = \frac{1}{4}\sum_{i,j} (p_i-p_j)^2 |G_{ij}|^2.
	\]
	For every term with $p_i+p_j>0$ and $i\neq j$, we have $p_i+p_j\le 1$ (since the eigenvalues sum to one), hence $2/(p_i+p_j)\ge 2$. The diagonal terms with $i=j$ vanish because $(p_i-p_i)^2=0$. Therefore
	\[
	F \ge 8\,\mathrm{HSS}^2.
	\]
	Equality $F = 8\,\mathrm{HSS}^2$ holds iff for every pair $(i,j)$ such that $(p_i-p_j)^2|G_{ij}|^2 \neq 0$, one has $p_i+p_j=1$. This condition is satisfied by all non-stationary pure states, as well as by all unitary qubit families (where the only contributing off-diagonal pair has $p_1+p_2=1$). More generally, in higher dimensions equality can hold for certain rank-$2$ mixed states, provided the generator $G$ couples only the two support eigenvectors so that every contributing pair obeys $p_i+p_j=1$.
	
	Using $r_{\mathrm{eff}} = \mathrm{TS}^2/\mathrm{HSS}^2$, we obtain the equivalent bounds:
	\[
	F \ge \frac{8}{r_{\mathrm{eff}}}\,\mathrm{TS}^2 \qquad \Longleftrightarrow \qquad
	r_{\mathrm{eff}} \ge \frac{8\,\mathrm{TS}^2}{F}.
	\]
	This inequality provides a lower bound on the effective rank. Whenever $8\,\mathrm{TS}^2/F$ substantially exceeds the universal minimum $2$, the effective rank is forced to be correspondingly large, indicating multi-mode dynamics. Conversely, a large QFI does not by itself force a large effective rank. The effective rank thus serves as a diagnostic: a modest metrological performance (small $F$) at a given level of state distinguishability (fixed $\mathrm{TS}$) signals that many singular modes are active.
	
	\subsection{Quantum speed limits with effective rank}
	\label{sec:qsl}
	
	A standard geometric bound on the trace distance follows from the triangle inequality for integrals \cite{DeffnerCampbell2017}:
	\[
	\begin{aligned}
		D(\rho(0),\rho(\tau)) 
		&= \frac12 \bigl\| \rho(\tau) - \rho(0) \bigr\|_1 \\
		&\le \frac12 \int_0^\tau \bigl\| \partial_\phi \rho(\phi) \bigr\|_1 \, d\phi 
		= \int_0^\tau \mathrm{TS}(\phi) \, d\phi .
	\end{aligned}
	\]
	This inequality is generally strict; saturation is possible only for specially structured paths. Defining the time-averaged trace speed $\overline{\mathrm{TS}} := \frac{1}{\tau}\int_0^\tau \mathrm{TS}(\phi)\,d\phi$ directly yields the quantum speed limit
	\[
	\tau \ge \frac{D(\rho(0),\rho(\tau))}{\overline{\mathrm{TS}}}.
	\]
	Using the relation $\mathrm{TS} = \sqrt{r_{\mathrm{eff}}}\,\mathrm{HSS}$, we obtain
	\[
	\tau \ge \frac{D}{\overline{\sqrt{r_{\mathrm{eff}}}\,\mathrm{HSS}}},
	\]
	where the overline denotes the same time average. Because $\mathrm{HSS}(\phi)\ge0$ along the path, the time average satisfies $\overline{\sqrt{r_{\mathrm{eff}}}\,\mathrm{HSS}} \le \sqrt{r_{\max}}\,\overline{\mathrm{HSS}}$, with $r_{\max} := \max_{0\le\phi\le\tau} r_{\mathrm{eff}}(\phi)$ the maximal effective rank during the evolution. Applying the universal bound $r_{\mathrm{eff}}\le d$ then gives the hierarchy
	\[
	\tau \ge \frac{D}{\sqrt{r_{\max}}\;\overline{\mathrm{HSS}}} \ge \frac{D}{\sqrt{d}\;\overline{\mathrm{HSS}}}.
	\]
	If the Hilbert-space dimension $d$ is odd, Theorem~\ref{thm:odd} sharpens the dimension-only factor to $\sqrt{d-1/d}$, yielding the improved bound
	\[
	\tau \ge \frac{D}{\sqrt{d-1/d}\;\overline{\mathrm{HSS}}}.
	\]
	For pure states, $r_{\mathrm{eff}}\equiv 2$ is constant, reducing the bound to the simple form
	\[
	\tau \ge \frac{D}{\sqrt{2}\;\overline{\mathrm{HSS}}}.
	\]
	These inequalities show explicitly how the effective rank controls the tightness of an HSS-based quantum speed limit relative to the TS-based one. In experiments, $\mathrm{HSS}$ is often easier to access than $\mathrm{TS}$; the effective rank quantifies exactly the factor by which the simpler bound is looser.
	
	\section{Sharp maximum in odd dimension}
	\label{sec:odd}
	
	When the Hilbert-space dimension $d$ is odd, the tracelessness constraint prevents exact balancing of the multiplicities of positive and negative eigenvalues at full rank, leading to a sharper global bound.
	
	\begin{theorem}
		\label{thm:odd}
		Let $d$ be odd, and let $X\neq0$ be a traceless Hermitian operator on $\mathbb C^d$. Then
		\[
		\frac{\|X\|_1}{\|X\|_2}\le \sqrt{d-\frac1d}.
		\]
		This bound is sharp. Equality holds if and only if $X$ has full rank and its positive eigenvalues are all equal, its negative eigenvalues are all equal, and the two multiplicities are $(d-1)/2$ and $(d+1)/2$, up to exchange.
	\end{theorem}
	
	\begin{proof}
		As in the proof of Theorem~\ref{thm:universal}, write the nonzero eigenvalues of $X$ as $\alpha_1,\dots,\alpha_m$, $-\beta_1,\dots,-\beta_n$ with $\alpha_k,\beta_j>0$, $m,n\ge1$, and $m+n=r\le d$. Set $Q = \sum_k\alpha_k = \sum_j\beta_j$. For fixed $Q$, $m$, and $n$, the denominator $\sum_k\alpha_k^2+\sum_j\beta_j^2$ is minimized when the $\alpha_k$ are all equal and the $\beta_j$ are all equal, leading to
		\begin{equation}
			\mathcal R(X)^2 \le \frac{4mn}{m+n}.
			\label{eq:mn_bound2}
		\end{equation}
		
		We now maximize the right-hand side over integers $m,n\ge1$ with $m+n\le d$. For a fixed sum $r=m+n$, the analysis in Eq.~\eqref{eq:fixedrmax} gives
		\begin{equation}
			\max_{m+n=r}\frac{4mn}{m+n}
			= \begin{cases}
				r, & r\ \text{even},\\[4pt]
				r-\dfrac{1}{r}, & r\ \text{odd}.
			\end{cases}
		\end{equation}
		Since $d$ is odd, the largest admissible odd rank is $r=d$, giving $d-1/d$, while the largest even rank is $r=d-1$, giving $d-1$. Because $d-1/d > d-1$ for any $d>1$, the global maximum is $d-1/d$; any smaller odd or even rank yields a strictly smaller value. Therefore
		\begin{equation}
			\mathcal R(X)^2 \le d-\frac{1}{d},
		\end{equation}
		which proves the bound.
		
		To see that the bound is attainable, choose $m=(d-1)/2$, $n=(d+1)/2$, and take all positive eigenvalues equal to $Q/m$, all negative eigenvalues equal to $-Q/n$. This saturates each inequality, giving $\mathcal R^2 = d-1/d$. By Lemma~\ref{lem:realization}, such an $X$ is locally realizable as the tangent of a density-operator family.
		
		The equality conditions follow from the chain of inequalities: the bound is saturated only if the positive and negative eigenvalues are constant on their supports, the rank is full ($r=d$), and the multiplicities take the extremal values $\{(d-1)/2,(d+1)/2\}$.
		
		Thus the odd-dimensional obstruction is purely arithmetic: full-rank traceless spectra cannot split into equally large positive and negative sectors when $d$ is odd.
	\end{proof}
	
	\section{Discussion}
	\label{sec:discussion}
	
	The ratio $\mathcal R$ is appealing because it separates the geometry of the singular-value profile of the tangent operator from the overall magnitude of the motion. If the tangent is rescaled by a common multiplicative factor, both speeds scale identically and the ratio remains unchanged. This gives a limited form of calibration robustness, but more general reconstruction biases need not cancel, so the ratio should not be described as universally self-calibrating.
	
	The most robust interpretation of $\mathcal R$ is as a local diagnostic of singular-value concentration: $\mathcal R^2$ is the participation ratio of the tangent singular-value distribution, behaving like an effective rank. Values above $\sqrt2$ quantify the participation of more than two singular modes and thereby provide a compact, basis-independent measure of how broadly the tangent spectrum is distributed. This does not uniquely identify the physical mechanism-mixedness, leakage, or dissipation may all increase $\mathcal R$-but it offers a rigorous starting point for further dynamical analysis.
	
	The decomposition into classical and quantum effective ranks provides a sharp separation of eigenvalue changes from unitary rotations, with the caveat that at degeneracies the splitting is basis-dependent. Equality in the bound $r_{\mathrm{eff}}\le r_C+r_Q$ is guaranteed when one of the two contributions vanishes. The inequalities linking $\mathcal R$ to the quantum Fisher information (QFI) and to quantum speed limits show that the speed ratio carries direct metrological significance. Specifically, the bound $r_{\mathrm{eff}} \ge 8\,\mathrm{TS}^2/F$ quantifies a trade-off: a limited QFI at fixed distinguishability forces the effective rank to be large when the bound substantially exceeds the universal minimum of $2$. The speed-limit hierarchy $\tau \ge D/(\sqrt{r_{\max}}\,\overline{\mathrm{HSS}})$ makes the role of the effective rank in bounding evolution times explicit.
	
	It would be interesting to determine which parts of this structure survive in generalized frameworks where the tangent operator is no longer Hermitian, although in such settings the interpretation of singular values becomes model dependent.
	
	\section{Conclusion}
	
	In summary, the trace-to-Hilbert-Schmidt speed ratio for differentiable quantum states is governed by the Hermitian traceless structure of density-operator tangents. This yields the sharp rank-dependent bounds $\sqrt2\le\mathcal R\le\sqrt r$ (even $r$) and $\sqrt2\le\mathcal R\le\sqrt{r-1/r}$ (odd $r$), identifies pure-state and qubit dynamics as rigidly pinned to the lower endpoint, and gives the sharp odd-dimensional maximum $\mathcal R\le\sqrt{d-1/d}$. The squared ratio $\mathcal R^2$ is the participation ratio of the tangent singular-value distribution, providing a natural effective-rank characterization of local quantum dynamics. We have further decomposed the effective rank into classical and quantum contributions, obtained the bound $r_{\mathrm{eff}}\le r_C+r_Q$ with equality when one contribution vanishes, and derived the inequality $F\ge 8\,\mathrm{TS}^2/r_{\mathrm{eff}}$ linking the effective rank to the quantum Fisher information. The resulting lower bound $r_{\mathrm{eff}}\ge 8\,\mathrm{TS}^2/F$ shows that a small QFI forces a correspondingly large effective rank when the bound is informative. Finally, the effective rank controls a family of quantum speed limits, $\tau \ge D/(\sqrt{r_{\max}}\,\overline{\mathrm{HSS}})$, demonstrating its utility in bounding evolution times. These results establish the speed ratio as a versatile tool that connects geometry, metrology, and quantum dynamics.

	\appendix
	
	\section{Explicit four-dimensional saturation example}
	
	For completeness, we present an explicit $d=4$ example that saturates the upper bound for even rank. Consider
	\begin{equation}
		G=\sigma_x\oplus\sigma_x
		=
		\begin{pmatrix}
			0&1&0&0\\
			1&0&0&0\\
			0&0&0&1\\
			0&0&1&0
		\end{pmatrix}
	\end{equation}
	and
	\begin{equation}
		\rho_0=
		\begin{pmatrix}
			3/8&0&0&0\\
			0&1/8&0&0\\
			0&0&3/8&0\\
			0&0&0&1/8
		\end{pmatrix}
		=
		\begin{pmatrix}
			3/8&0\\
			0&1/8
		\end{pmatrix}
		\oplus
		\begin{pmatrix}
			3/8&0\\
			0&1/8
		\end{pmatrix}.
	\end{equation}
	For the unitary family
	\begin{equation}
		\rho(\phi)=e^{-i\phi G}\rho_0 e^{i\phi G},
	\end{equation}
	the tangent at $\phi=0$ is
	\begin{equation}
		\partial_\phi\rho\big|_{\phi=0}=-i[G,\rho_0].
	\end{equation}
	Since the matrices are block diagonal, it suffices to evaluate one block. Let
	\begin{equation}
		G_1=
		\begin{pmatrix}
			0&1\\
			1&0
		\end{pmatrix},
		\qquad
		\rho_1=
		\begin{pmatrix}
			3/8&0\\
			0&1/8
		\end{pmatrix}.
	\end{equation}
	Then
	\begin{equation}
		G_1\rho_1=
		\begin{pmatrix}
			0&1/8\\
			3/8&0
		\end{pmatrix},
		\qquad
		\rho_1G_1=
		\begin{pmatrix}
			0&3/8\\
			1/8&0
		\end{pmatrix},
	\end{equation}
	so that
	\begin{equation}
		[G_1,\rho_1]
		=
		\begin{pmatrix}
			0&-1/4\\
			1/4&0
		\end{pmatrix},
		\qquad
		-i[G_1,\rho_1]
		=
		\begin{pmatrix}
			0&i/4\\
			-i/4&0
		\end{pmatrix}.
	\end{equation}
	This Hermitian block has eigenvalues $\pm1/4$, and therefore singular values $1/4$ and $1/4$. The full tangent operator consists of two identical blocks, so its four nonzero singular values are all equal to $1/4$. It follows that
	\begin{equation}
		\norm{\partial_\phi\rho}_1=1,
		\qquad
		\norm{\partial_\phi\rho}_2=\frac12,
	\end{equation}
	and hence
	\begin{equation}
		\mathrm{TS}=\frac12,
		\qquad
		\mathrm{HSS}=\frac14,
		\qquad
		\mathcal R=2.
	\end{equation}
	Since the tangent rank is $r=4$, this realizes
	\begin{equation}
		\mathcal R=\sqrt4=\sqrt r,
	\end{equation}
	thus saturating the even-rank upper bound.

\end{document}